\documentclass{amsart}

\usepackage{amssymb}

\usepackage{amssymb,amsmath}

\usepackage[english]{babel}

\newtheorem{theorem}{Theorem}[section]
\newtheorem{lemma}[theorem]{Lemma}
\newtheorem{proposition}[theorem]{Proposition}
\newtheorem{corollary}[theorem]{Corollary}

\numberwithin{equation}{section}

\begin{document}

\title{Symmetries of the Black-Scholes equation}

\author{{\bf Paul Lescot \rm}}
\address{Laboratoire de Math\'ematiques Rapha\"el Salem \\
UMR 6085 CNRS \\ 
Universit\'e de Rouen \\
Technop\^ole du Madrillet \\
Avenue de l'Universit\'e, B.P. 12 \\
76801 Saint-Etienne-du-Rouvray (FRANCE) \\ 
Phone 00 33 (0)2 32 95 52 24 \\
Fax 00 33 (0)2 32 95 52 86 \\
Paul.Lescot@univ-rouen.fr \\}

\date{January 14, 2011}

\setcounter{section}{0}

\begin{abstract}

We determine the algebra of isovectors for the Black--Scholes equation.
As a consequence, we obtain some previously unknown transformations on
the solutions.

MSC 34 A 26, 91 B 28

\end{abstract}

\maketitle

\setcounter{section}{0}

\numberwithin{equation}{section}

\section{Introduction}
We study the algebro--geometrical structure of the Black--Scholes equation.

  The computation of the symmetry group of a partial differential equation 
has proven itself to be a very useful tool in Classical Mathematical Physics (see \it e.g. \rm \cite{2}), as well as in Euclidean Quantum Mechanics
(see \it e.g. \rm \cite{3} and \cite{4}).
   It was therefore natural to try and use the same method in Financial Mathematics. 
   
After setting the general framework (\S 2), and performing some preliminary reductions (\S 3), we determine (\S 4)
the isovectors for the Black-Scholes equation in a way broadly similar to the one used for the backward heat equation with potential term in the second of the two aforementioned joint papers with J.-C. Zambrini. Our computation turns out to suggest Black and Scholes' original solution method (\cite{1}) of their equation; in particular, the quantities $r-\frac{{\sigma}^{2}}{2}$ and $r+\frac{{\sigma}^{2}}{2}$ appear naturally in this context. As  corollaries, we determine (\S 5) the structure of the Lie algebra of the symmetry group of
the equation, then (\S 6) we obtain some interesting transformations on the solutions.

\newpage

\section{Generalities and notations}

We shall be concerned with the classical Black-Scholes equation : 

$$
\displaystyle\frac{\partial C}{\partial t}+\displaystyle\frac{1}{2}\sigma^{2}S^{2}\displaystyle\frac{\partial^{2}C}{\partial S^{2}}
+rS\displaystyle\frac{\partial C}{\partial S}-rC=0\,\, (\mathcal E) 
$$
for the price $C(t,S)$ of a call option with maturity $T$ and strike price $K$
on an underlying asset satisfying $S_{t}=S$ (see \cite{1}, where $\sigma$ is denoted by $v$, $C$ by $w$, and $S$ by $x$).
As is well--known (\cite{1}, p.646), the same equation is satisfied by the price of a
put option.

We assume $\sigma>0$, and define
$$
\tilde{r}:=r-\displaystyle\frac{\sigma^{2}}{2}\,\, 
$$
and
$$
\tilde{s}:=r+\displaystyle\frac{\sigma^{2}}{2}\,\, .
$$
It is useful to remark that
$$
\displaystyle\frac{\tilde{r}^{2}}{2\sigma^{2}}+r=\displaystyle\frac{\tilde{s}^{2}}{2\sigma^{2}}\,\, .
$$
We intend to determine the isovectors for $(\mathcal E)$, using the method applied,
in \cite{2}, pp. 657--658 (see also \cite{3}, pp.189--192) to the heat equation, and in \cite{4}, \S 3, to the (backward) heat equation with a potential term.

Let us set $x=\ln(S)$ and
$$
\varphi(t,x):=C(t,e^{x})=C(t,S)\,\, ;
$$ 
then
$\varphi$ is defined on $\mathbf R_{+}\times \mathbf R$.
One has
$$
\displaystyle\frac{\partial C}{\partial S}=\displaystyle\frac{1}{S}\displaystyle\frac{\partial \varphi}{\partial x}\,\, ,
$$
$$
\displaystyle\frac{\partial^{2} C}{\partial S^{2}}=-\displaystyle\frac{1}{S^{2}}\displaystyle\frac{\partial \varphi}{\partial x}+\displaystyle\frac{1}{S^{2}}\displaystyle\frac{\partial^{2}\varphi}{\partial x^{2}}
$$
and
$$
\displaystyle\frac{\partial C}{\partial t}=\displaystyle\frac{\partial \varphi}{\partial t}\,\, .
$$
Equation $(\mathcal E)$ is therefore equivalent to the following equation in $\varphi$ :
$$
\displaystyle\frac{\partial \varphi}{\partial t}+\displaystyle\frac{\sigma^{2}}{2}(\displaystyle\frac{\partial^{2}\varphi}{\partial x^{2}}-\displaystyle\frac{\partial \varphi}{\partial x})+r\displaystyle\frac{\partial \varphi}{\partial x}-r\varphi=0 \,\,  (\mathcal E_{1}) \,\, ,
$$
that is :
$$
\displaystyle\frac{\partial \varphi}{\partial t}+\displaystyle\frac{\sigma^{2}}{2}\displaystyle\frac{\partial^{2}\varphi}{\partial x^{2}}+\tilde{r}\displaystyle\frac{\partial \varphi}{\partial x}-r\varphi=0 \,\,  (\mathcal E_{2}) \,\, .
$$

\newpage

\section{Computation of the isovectors : preliminary reductions}

Let us set $A=\displaystyle\frac{\partial \varphi}{\partial x}$ and $B=\displaystyle\frac{\partial \varphi}{\partial t}$,
and consider thenceforth $t$, $x$, $\varphi$, $A$ and $B$ as \it independent \rm variables . Then $(\mathcal E_{2})$ is equivalent to the vanishing, on the five--dimensional manifold $M=\mathbf R_{+} \times \mathbf R^{4} $ of $(t,x,\varphi,A,B)$, of the following
system of differential forms :
\begin{eqnarray}
\alpha=d\varphi-Adx-Bdt\,\, ,
\end{eqnarray}

\begin{eqnarray}
d\alpha=-dAdx-dBdt \,\, ,
\end{eqnarray}
and
\begin{eqnarray}
\beta=(B+\tilde{r}A-r\varphi)dxdt+\displaystyle\frac{1}{2}\sigma^{2}dAdt \,\, .
\end{eqnarray}

Let $I$ denote the ideal of $\Lambda T^{*}(M)$ generated by $\alpha$, $d\alpha$ and $\beta$ ; as
\begin{eqnarray}
d\beta 
&=&(dB+\tilde{r}dA-rd\varphi)dxdt \nonumber \\
&=&d\alpha (dx-\tilde{r}dt)+\alpha (-rdxdt)\in I \,\, ,  
\end{eqnarray}
$I$ is a differential ideal of $\Lambda T^{*}(M)$.
By definition (see \cite{2}), an isovector for $(\mathcal E_{2})$ is a vector field
\begin{eqnarray}
N=N^{t}\displaystyle\frac{\partial}{\partial t}+N^{x}\displaystyle\frac{\partial}{\partial x}+N^{\varphi}\displaystyle\frac{\partial}{\partial \varphi}+N^{A}\displaystyle\frac{\partial}{\partial A}+N^{B}\displaystyle\frac{\partial}{\partial B}
\end{eqnarray}
such that
\begin{eqnarray}
\mathcal L_{N}(I)\subseteq I\,\, .
\end{eqnarray}
Using the formal properties of the Lie derivative
(\cite{2}, p.654), one easily proves that the set $\mathcal G$ of these isovectors constitutes a Lie algebra
(for the usual bracket of vector fields).

In order to determine $\mathcal G$, we may use a trick first explained in 

\cite{2}, p.657, that applies
in all situations in which there is only one $1$--form among the given generators of the ideal $I$
(see also \cite{4}, p.211). 

Let $N\in \mathcal G$ ; as $\mathcal L_{N}(I)\subseteq I$,
one has $\mathcal L_{N}(\alpha)\in I=<\alpha,d\alpha,\beta>$, whence there is a $0$--form
(\it i.e. \rm a function) $\lambda$ such that $\mathcal L_{N}(\alpha)=\lambda\alpha$.
Let us define
\begin{eqnarray}
F:=N \rfloor \alpha=N^{\varphi}-AN^{x}-BN^{t}\,\, .
\end{eqnarray}
This can be rewritten as
\begin{eqnarray}
\lambda\alpha=\mathcal L_{N}(\alpha)=N\rfloor d\alpha+d(N\rfloor\alpha)=N \rfloor d\alpha+dF \,\, ,
\end{eqnarray}
whence
\begin{eqnarray}
N\rfloor d\alpha=\lambda\alpha-dF \,\, ,
\end{eqnarray}
that is
\begin{eqnarray}
N\rfloor(-dAdx-dBdt)=\lambda\alpha - dF 
\end{eqnarray}
\it i.e. \rm
\begin{eqnarray}
-N^{A}dx+N^{x}dA-N^{B}dt+N^{t}dB \\
&=&\lambda\alpha-dF \nonumber \\
&=&\lambda(d\varphi - Adx -Bdt)-dF \,\, . \nonumber 
\end{eqnarray}
Whence (letters as lower indices indicating differentiation,
as usual)

$$
(*)\left\{
\begin{array}{lr}
-N^{B}=-\lambda B-F_{t} \nonumber \\
-N^{A}=-\lambda A-F_{x} \nonumber \\
0=\lambda-F_{\varphi} \nonumber \\
N^{x}=-F_{A} \nonumber \\
N^{t}=-F_{B} \,\, . \nonumber 
\end{array}
\right.
$$

Using the third equation, we can eliminate $\lambda$ and obtain

$$
(**)\left\{
\begin{array}{lr}
N^{t}=-F_{B} \\
N^{x}=-F_{A}  \\
N^{\varphi}=F-AF_{A}-BF_{B}  \\
N^{A}=F_{x}+AF_{\varphi} \\
N^{B}=F_{t}+BF_{\varphi}  \,\, .
\end{array}
\right.
$$

Conversely, the existence of a function $F(t,x,\varphi,A,B)$ such that the above equations hold
clearly implies that $\mathcal L_{N}(\alpha)=F_{\varphi}\alpha\in I$ ;
but then $$\mathcal L_{N}(d\alpha)=d(\mathcal L_{N}(\alpha))\in d(I) \subseteq I\,\, ,$$
and there only remains to be satisfied the condition $$\mathcal L_{N}(\beta)\in I\,\, .$$

\section{The general isovector}
The last condition in the previous paragraph can be stated as
\begin{eqnarray}
\mathcal L_{N}(\beta)=\rho\alpha+\xi d\alpha+\omega \beta  \,\, ,
\end{eqnarray}
for $\rho$ a $1$--form, $\xi$ a $0$--form and $\omega$ a $0$--form.
Let $D$ denote the coefficient of $d\varphi$ in $\rho$ ;
replacing $\rho$ by $\rho-D\alpha$ (which doesn't affect the validity of
$(4.1)$ as $\alpha^{2}=0$), we may assume that $D=0$.
Setting then 
\begin{eqnarray}
\rho=R_{1}dt+R_{2}dx+R_{3}dA+R_{4}dB \,\, ,
\end{eqnarray}
\begin{eqnarray}
\xi=R_{5} \,\, ,
\end{eqnarray}
and
\begin{eqnarray}
\omega=R_{6} \,\, ,
\end{eqnarray}
we shall obtain a system of ten equations in $F$ ; we shall then eliminate $R_{1}$,...,$R_{6}$.

Identifying the coefficients of, in that order, $dtdx$, $dtd\varphi$, $dtdA$, $dtdB$, $dxd\varphi$,
$dxdA$, $dxdB$, $d\varphi dA$, $d\varphi dB$ and $dAdB$ yields the following 
system :

\begin{eqnarray}
&&rN^{\varphi}-\tilde{r}N^{A}-N^{B}+(r\varphi-\tilde{r}A-B)N_{x}^{x} 
+(r\varphi-\tilde{r}A-B)N_{t}^{t}-\displaystyle\frac{1}{2}\sigma^{2}N_{x}^{A} \nonumber \\
&=&-AR_{1}+BR_{2}-(B+\tilde{r}A-r\varphi)R_{6}  
\end{eqnarray}

\begin{eqnarray}
(r\varphi-\tilde{r}A-B)N_{\varphi}^{x}-\displaystyle\frac{1}{2}\sigma^{2}N_{\varphi}^{A}=R_{1}
\end{eqnarray}

\begin{eqnarray}
(r\varphi-\tilde{r}A-B)N_{A}^{x}-\displaystyle\frac{1}{2}\sigma^{2}N_{A}^{A}-\displaystyle\frac{1}{2}\sigma^{2}N_{t}^{t}
=BR_{3}-\displaystyle\frac{1}{2}\sigma^{2}R_{6}
\end{eqnarray}

\begin{eqnarray}
(r\varphi-\tilde{r}A-B)N_{B}^{x}-\displaystyle\frac{1}{2}\sigma^{2}N_{B}^{A}
=BR_{4}+R_{5}
\end{eqnarray}

\begin{eqnarray}
(B+\tilde{r}A-r\varphi)N_{\varphi}^{t}=R_{2}
\end{eqnarray}

\begin{eqnarray}
(B+\tilde{r}A-r\varphi)N_{A}^{t}-\displaystyle\frac{1}{2}\sigma^{2}N_{x}^{t}
=AR_{3}+R_{5}
\end{eqnarray}

\begin{eqnarray}
(B+\tilde{r}A-r\varphi)N_{B}^{t}=AR_{4}
\end{eqnarray}

\begin{eqnarray}
-\displaystyle\frac{1}{2}\sigma^{2}N_{\varphi}^{t}
=-R_{3}
\end{eqnarray}

\begin{eqnarray}
0=-R_{4} 
\end{eqnarray}

\begin{eqnarray}
\displaystyle\frac{1}{2}\sigma^{2}N_{B}^{t}=0 \,\, . 
\end{eqnarray}

Equation $(4.13)$ gives $R_{4}=0$ ; then $(4.8)$ defines $R_{5}$.
Equation $(4.14)$ is equivalent to $N_{B}^{t}=0$ ; if that be the case, then $(4.11)$ holds
automatically. Now $(4.9)$ defines $R_{2}$, $(4.12)$ defines $R_{3}$, $(4.6)$ defines $R_{1}$ and $(4.7)$ defines $R_{6}$.
We are left with equations $(4.5)$ and $(4.10)$ and the condition $N_{B}^{t}=0$.

Let us begin with the last mentioned ; it is equivalent to $F_{BB}=0$ : $F$ is affine in $B$,
\it i.e. \rm $F=c+Bd$, where $c$ and $d$ depend only on $(t,x,\varphi,A)$.
Now $(**)$ can be rewritten as

$$
(***)\left\{
\begin{array}{lr}
N^{t}=-d \nonumber \\
N^{x}=-c_{A}-Bd_{A} \nonumber \\
N^{\varphi}=c-Ac_{A}-ABd_{A} \nonumber \\
N^{A}=c_{x}+Bd_{x}+Ac_{\varphi}+ABd_{\varphi}\nonumber \\
N^{B}=c_{t}+Bd_{t}+Bc_{\varphi}+B^{2}d_{\varphi} \,\, .\nonumber 
\end{array}
\right.
$$
From $(4.12)$ follows
\begin{eqnarray}
R_{3}=-\displaystyle\frac{1}{2}\sigma^{2}d_{\varphi} \,\, ,
\end{eqnarray}
and $(4.8)$ yields
\begin{eqnarray}
R_{5}=(B+\tilde{r}A-r\varphi)d_{A}-\displaystyle\frac{1}{2}\sigma^{2}(d_{x}+Ad_{\varphi})\,\, .
\end{eqnarray}
Now we can rewrite $(4.10)$ as
\begin{eqnarray}
\\
&& (B+\tilde{r}A-r\varphi)(-d_{A})-\displaystyle\frac{1}{2}\sigma^{2}(-d_{x}) \nonumber \\
&=&(B+\tilde{r}A-r\varphi)d_{A}-\displaystyle\frac{1}{2}\sigma^{2}(d_{x}+Ad_{\varphi})  \,\, .\nonumber  
\end{eqnarray}
Comparing the coefficients of $B$ on both sides gives
$-d_{A}=d_{A}$, that is $d_{A}=0$, \it i.e. \rm $d$ depends only upon
$(t,x,\varphi)$.
Now $(4.17)$ becomes 
\begin{eqnarray}
\displaystyle\frac{1}{2}\sigma^{2}d_{x}=-\displaystyle\frac{1}{2}\sigma^{2}(d_{x}+Ad_{\varphi})
\end{eqnarray}
that is
\begin{eqnarray}
2d_{x}=-Ad_{\varphi}\,\, .
\end{eqnarray}
Differentiating with respect to $A$ leads to
\begin{eqnarray}
0=2d_{Ax}=-d_{\varphi} \,\, ,
\end{eqnarray}
whence
$d_{\varphi}=0$ and $d_{x}=-\displaystyle\frac{1}{2}Ad_{\varphi}=0$ : $d$ depends only upon $t$.
Then $(***)$ becomes
$$
(****)\left\{
\begin{array}{lr}
N^{t}=-d \nonumber \\
N^{x}=-c_{A} \nonumber \\
N^{\varphi}=c-Ac_{A} \nonumber \\
N^{A}=c_{x}+Ac_{\varphi}\nonumber \\
N^{B}=c_{t}+Bd_{t}+Bc_{\varphi} \,\, ,\nonumber \\
\end{array}
\right.
$$
the new unknowns being a function $c(t,x,\varphi,A)$ and a function $d(t)$,
and we still have to satisfy equation $(4.5)$.
Now $(4.9)$ implies $R_{2}=0$, and $(4.6)$ gives
\begin{eqnarray}
R_{1}=-(r\varphi-\tilde{r}A-B)c_{A\varphi}-\displaystyle\frac{1}{2}\sigma^{2}(c_{x\varphi}+Ac_{\varphi \varphi})\,\, .
\end{eqnarray}
Equation $(4.7)$ now becomes
\begin{eqnarray}
\nonumber \\
&&(r\varphi-\tilde{r}A-B)(-c_{AA})-\displaystyle\frac{1}{2}\sigma^{2}(c_{xA}+c_{\varphi}+Ac_{A\varphi})-\displaystyle\frac{1}{2}\sigma^{2}(-d_{t}) \nonumber \\
&=&-\displaystyle\frac{1}{2}\sigma^{2}R_{6} \,\, , \nonumber \\
\end{eqnarray}
that is :
\begin{eqnarray}
R_{6}=-d_{t}+c_{xA}+c_{\varphi}+Ac_{A\varphi}+\displaystyle\frac{2}{\sigma^{2}}(B+\tilde{r}A--r\varphi)(-c_{AA})\,\, .
\end{eqnarray}
Eliminating $R_{1}$, $R_{2}$ and $R_{6}$ turns $(4.5)$ into
\begin{eqnarray}
\nonumber \\
&&r(c-Ac_{A})-\tilde{r}(c_{x}+Ac_{\varphi})-(c_{t}+Bd_{t}+Bc_{\varphi})+(r\varphi-\tilde{r}A-B)(-c_{Ax})  \nonumber \\
&+&(r\varphi-\tilde{r}A-B)(-d_{t})-\displaystyle\frac{1}{2}\sigma^{2}(c_{xx}+Ac_{\varphi x}) \nonumber \\
&=&-A(-(r\varphi-\tilde{r}A-B)c_{A,\varphi}-\displaystyle\frac{1}{2}\sigma^{2}(c_{x\varphi}+Ac_{\varphi \varphi}))  \\
&&-(B+\tilde{r}A-r\varphi)(-d_{t}+c_{xA}+c_{\varphi}+Ac_{A\varphi}+\displaystyle\frac{2}{\sigma^{2}}(B+\tilde{r}A-r\varphi)(-c_{AA})) \nonumber \,\, .
\end{eqnarray}
Both members of $(4.24)$ are second--order polynomials in $B$ ; equating the coefficients of
$B^{2}$ gives $c_{AA}=0$, whence $c$ is affine in $A$ : $c=e+Af$ with $e$ and $f$ functions of
$(t,x,\varphi)$.

Equating the coefficients of $B$ yields
\begin{eqnarray}
-d_{t}-c_{\varphi}+c_{Ax}+d_{t}=-Ac_{A\varphi}-(-d_{t}+c_{Ax}+c_{\varphi}+Ac_{A\varphi})
\end{eqnarray}
that is :
\begin{eqnarray}
d_{t}=2c_{Ax}+2Ac_{A\varphi}=2f_{x}+2Af_{\varphi}\,\, .
\end{eqnarray}
Differentiating the last equality with respect to $A$ gives $f_{\varphi}=0$
(that is, $f$ is a function of $(t,x)$), and then we get that $d_{t}=2f_{x}$,
\it i.e. \rm
\begin{eqnarray}
f=\displaystyle\frac{1}{2}d^{'}(t)x+\mu (t)
\end{eqnarray}
for some function $\mu$ of $t$ alone.

Equating the constant terms in $B$ now gives us

\begin{eqnarray}
&&re-\tilde{r}(e_{x}+Af_{x}+Ae_{\varphi})-(e_{t}+Af_{t})+(r\varphi-\tilde{r}A)(-f_{x}) \nonumber \\
&&+(r\varphi-\tilde{r}A)(-d_{t})-\displaystyle\frac{1}{2}\sigma^{2}(e_{xx}+Af_{xx}+Ae_{\varphi x}) \nonumber \\
&=&A\displaystyle\frac{1}{2}\sigma^{2}(e_{x\varphi}+Ae_{\varphi \varphi})
-(\tilde{r}A-r\varphi)(-d_{t}+f_{x}+e_{\varphi}) \nonumber \,\,  . \\
\end{eqnarray}

Both sides of equation $(4.28)$ are polynomials in $A$ with coefficients depending only upon
$(t,x,\varphi)$. Identifying the terms in $A^{2}$ leads us to $e_{\varphi \varphi }=0$,
whence $e=g+h\varphi$ with $g$ and $h$ depending only upon $(t,x)$.
The unknowns are now $d$ and $\mu$ (functions of $t$ alone) and $g$ and $h$
(functions of $(t,x)$).

Identifying the coefficients of $A$ on both sides of $(4.29)$ yields
\begin{eqnarray}
\nonumber \\
&&-\tilde{r}(f_{x}+e_{\varphi})-f_{t}+\tilde{r}f_{x}
+\tilde{r}d_{t}-\displaystyle\frac{1}{2}\sigma^{2}(f_{xx}+e_{\varphi x}) \nonumber \\
&=&\displaystyle\frac{1}{2}\sigma^{2}e_{x\varphi}
-\tilde{r}(-d_{t}+f_{x}+e_{\varphi}) \,\, , \nonumber \\
\end{eqnarray}
that is 
\begin{eqnarray}
-f_{t}-\displaystyle\frac{1}{2}\sigma^{2}f_{xx}
=\sigma^{2}h_{x}
-\tilde{r}f_{x}
\end{eqnarray}
or
\begin{eqnarray}
\sigma^{2}h_{x}=\tilde{r}\displaystyle\frac{1}{2}d^{'}(t)-\displaystyle\frac{1}{2}d^{''}(t)x-\mu^{'}(t) \,\, ,
\end{eqnarray}
that is
\begin{eqnarray}
h(t,x)=\displaystyle\frac{\tilde{r}}{2\sigma^{2}}d^{'}(t)x-\displaystyle\frac{d^{''}(t)}{4\sigma^{2}}x^{2}-\displaystyle\frac{\mu^{'}(t)}{\sigma^{2}}x+k(t) \,\, ,
\end{eqnarray}
for some function $k$ of $t$ alone.

The constant term gives
\begin{eqnarray}
re-\tilde{r}e_{x}-e_{t}-r\varphi f_{x}
-r\varphi d_{t}-\displaystyle\frac{1}{2}\sigma^{2}e_{xx}
=r\varphi(-d_{t}+f_{x}+e_{\varphi})\,\, .
\end{eqnarray}
According to the relation $d_{t}=2f_{x}$,
this becomes
\begin{eqnarray}
rg-\tilde{r}g_{x}-\tilde{r}\varphi h_{x}-g_{t}-\varphi h_{t}
-r\varphi d_{t}-\displaystyle\frac{1}{2}\sigma^{2}g_{xx}
-\displaystyle\frac{1}{2}\sigma^{2}\varphi h_{xx}=0 \,\, . 
\end{eqnarray}

Equating the constant term (in $\varphi$) of $(4.34)$ to zero gives that $g$ is a solution 
of $(\mathcal E_{2})$.

The value of the term in $\varphi$ means that
\begin{eqnarray}
-\tilde{r}h_{x}-h_{t}
-rd_{t}
-\displaystyle\frac{1}{2}\sigma^{2}h_{xx}=0 \,\, ,
\end{eqnarray}
that is
\begin{eqnarray}
&&-\tilde{r}(\displaystyle\frac{\tilde{r}}{2\sigma^{2}}d^{'}(t)-\displaystyle\frac{d^{''}(t)}{2\sigma^{2}}x-\displaystyle\frac{\mu^{'}(t)}{\sigma^{2}}) \nonumber \\ &&-(\displaystyle\frac{\tilde{r}}{2\sigma^{2}}d^{''}(t)x-\displaystyle\frac{d^{'''}(t)}{4\sigma^{2}}x^{2}-\displaystyle\frac{\mu^{''}(t)}{\sigma^{2}}x+k^{'}(t))
-rd^{'}(t)
-\displaystyle\frac{1}{2}\sigma^{2}(-\displaystyle\frac{d^{''}(t)}{2\sigma^{2}})=0 . \nonumber \\
\end{eqnarray}

This is a polynomial equation in $x$ with functions of $t$ as coefficients.
Considering the coefficient of $x^{2}$ gives $d^{'''}(t)=0$,
that is $d(t)=C_{1}t^{2}+C_{2}t+C_{3}$ for constants $C_{1}$, $C_{2}$ and $C_{3}$.

Equating the terms in $x$ leads to
\begin{eqnarray}
\displaystyle\frac{\tilde{r}}{2\sigma^{2}}d^{''}(t)-\displaystyle\frac{\tilde{r}}{2\sigma^{2}}d^{''}(t)+\displaystyle\frac{\mu^{''}(t)}{\sigma^{2}}=0 \,\, ,
\end{eqnarray}
that is
\begin{eqnarray}
\mu^{''}(t)=0
\end{eqnarray}
or
\begin{eqnarray}
\mu(t)=C_{4}t+C_{5}
\end{eqnarray}
($C_{4}$, $C_{5}$ real constants).

We are left with the constant term in $x$ :
\begin{eqnarray}
-\displaystyle\frac{\tilde{r}^{2}}{2\sigma^{2}}d^{'}(t)+\tilde{r}\displaystyle\frac{\mu^{'}(t)}{\sigma^{2}}-k^{'}(t)
-rd^{'}(t)+\displaystyle\frac{d^{''}(t)}{4}=0 \,\, ,
\end{eqnarray}
or
\begin{eqnarray}
k(t) 
&=&-(\displaystyle\frac{\tilde{r}^{2}}{2\sigma^{2}}+r)d(t)+\tilde{r}\displaystyle\frac{\mu(t)}{\sigma^{2}}
+\displaystyle\frac{d^{'}(t)}{4}+C_{6} \nonumber \\
&=&-\displaystyle\frac{\tilde{s}^{2}}{2\sigma^{2}}d(t)+\tilde{r}\displaystyle\frac{\mu(t)}{\sigma^{2}}
+\displaystyle\frac{d^{'}(t)}{4}+C_{6}  \nonumber \\
\end{eqnarray}
for some constant $C_{6}$.
Therefore
\begin{eqnarray}
f(t,x) 
&=&\displaystyle\frac{1}{2}d^{'}(t)x+\mu (t) \nonumber \\
&=&\displaystyle\frac{1}{2}x(2C_{1}t+C_{2})+C_{4}t+C_{5} \nonumber \\
\end{eqnarray}
and
\begin{eqnarray}
h(t,x)  
&=&\displaystyle\frac{\tilde{r}}{2\sigma^{2}}x(2C_{1}t+C_{2})-\displaystyle\frac{C_{1}}{2\sigma^{2}}x^{2} \nonumber \\
&-&\displaystyle\frac{C_{4}x}{\sigma^{2}}-\displaystyle\frac{\tilde{s}^{2}}{2\sigma^{2}}(C_{1}t^{2}+C_{2}t+C_{3}) \nonumber \\
&+&\tilde{r}\displaystyle\frac{(C_{4}t+C_{5})}{\sigma^{2}}
+\displaystyle\frac{2C_{1}t+C_{2}}{4}+C_{6} \,\, .\nonumber \\
\end{eqnarray}

We have established
\begin{theorem}
The general isovector $N$ for $(\mathcal E_{2})$ is given, in terms of an arbitrary solution $g$ of $(\mathcal E_{2})$
and six arbitrary real constants $C_{1}$,...,$C_{6}$, by the formulas
$$
N^{t}=-C_{1}t^{2}-C_{2}t-C_{3}
$$
$$
N^{x}=-\displaystyle\frac{1}{2}x(2C_{1}t+C_{2})-(C_{4}t+C_{5})
$$
\begin{eqnarray}
N^{\varphi} \nonumber
&=&g+\varphi(\displaystyle\frac{\tilde{r}}{2\sigma^{2}}x(2C_{1}t+C_{2})-\displaystyle\frac{C_{1}}{2\sigma^{2}}x^{2}-\displaystyle\frac{C_{4}x}{\sigma^{2}} \nonumber \\
&&-\displaystyle\frac{\tilde{s}^{2}}{2\sigma^{2}}(C_{1}t^{2}+C_{2}t+C_{3})+\tilde{r}\displaystyle\frac{C_{4}t+C_{5}}{\sigma^{2}}
+\displaystyle\frac{2C_{1}t+C_{2}}{4}+C_{6}) \nonumber 
\end{eqnarray}
\begin{eqnarray}
N^{A} \nonumber
&=&g_{x}+\varphi(\displaystyle\frac{\tilde{r}}{2\sigma^{2}}(2C_{1}t+C_{2})-\displaystyle\frac{C_{1}x}{\sigma^{2}}-\displaystyle\frac{C_{4}}{\sigma^{2}}) \nonumber \\
&+&\displaystyle\frac{A}{2}(2C_{1}t+C_{2}) \nonumber \\
&+&A(\displaystyle\frac{\tilde{r}}{2\sigma^{2}}x(2C_{1}t+C_{2})-\displaystyle\frac{C_{1}}{2\sigma^{2}}x^{2} \nonumber \\
&-&\displaystyle\frac{C_{4}x}{\sigma^{2}}-\displaystyle\frac{\tilde{s}^{2}}{2\sigma^{2}}(C_{1}t^{2}+C_{2}t+C_{3}) \nonumber \\
&+&\tilde{r}\displaystyle\frac{(C_{4}t+C_{5})}{\sigma^{2}}
+\displaystyle\frac{2C_{1}t+C_{2}}{4}+C_{6}) .\nonumber 
\end{eqnarray}

\begin{eqnarray}
N^{B} \nonumber
&=&g_{t}+\varphi(\displaystyle\frac{\tilde{r}x}{\sigma^{2}}-\displaystyle\frac{\tilde{s}^{2}}{2\sigma^{2}}(2C_{1}t+C_{2})+\displaystyle\frac{\tilde{r}C_{4}}{\sigma^{2}}+\displaystyle\frac{C_{1}}{2}) \nonumber \\
&+&A(xC_{1}+C_{4})+B(2C_{1}t+C_{2}) \nonumber \\
&+&B(\displaystyle\frac{\tilde{r}}{2\sigma^{2}}x(2C_{1}t+C_{2})-\displaystyle\frac{C_{1}}{2\sigma^{2}}x^{2} \nonumber \\
&-&\displaystyle\frac{C_{4}x}{\sigma^{2}}-\displaystyle\frac{\tilde{s}^{2}}{2\sigma^{2}}(C_{1}t^{2}+C_{2}t+C_{3}) \nonumber \\
&+&\tilde{r}\displaystyle\frac{(C_{4}t+C_{5})}{\sigma^{2}}
+\displaystyle\frac{2C_{1}t+C_{2}}{4}+C_{6}) .\nonumber 
\end{eqnarray}
\end{theorem}

\section{The Lie algebra}

For $u$ a solution of $(\mathcal E_{2})$, let $N_{u}$ denote the isovector defined
by $g=u$ and $C_{1}=...=C_{6}=0$, and, for $1\leq i \leq 6$, let $N_{i}$ denote the isovector
defoned by by $g=0$, $C_{i}=1$ and $C_{j}=0$ for $j\neq i$.

The function $g$ and $h$ are determined by the isovector $N$ :
$$
g=N^{\varphi}-\varphi \displaystyle\frac{\partial N^{\varphi}}{\partial \varphi} 
$$
and
$$
h=\displaystyle\frac{\partial N^{\varphi}}{\partial \varphi} \,\, ;
$$
we shall denote them repectively by $g_{N}$ and $h_{N}$.
As seen in \S 4, $h_{N}$ is a function of $(t,x)$.

\begin{lemma}
For all $(M,N)\in \mathcal G^{2}$,
$$
g_{[M,N]}=M^{t}\displaystyle\frac{\partial g_{N}}{\partial t}+M^{x}\displaystyle\frac{\partial g_{N}}{\partial x}+g_{M}h_{N}
-N^{t}\displaystyle\frac{\partial g_{M}}{\partial t}-N^{x}\displaystyle\frac{\partial g_{M}}{\partial x}-g_{N}h_{M}
$$
and
$$
h_{[M,N]}= M^{t}\displaystyle\frac{\partial h_{N}}{\partial t}+M^{x}\displaystyle\frac{\partial h_{N}}{\partial x}-N^{t}\displaystyle\frac{\partial h_{M}}{\partial t}-N^{x}\displaystyle\frac{\partial h_{M}}{\partial x} \,\, .
$$
\end{lemma}
\begin{proof}
One has\begin{eqnarray}
[M,N]^{\varphi} 
&=&M(N^{\varphi})-N(M^{\varphi}) \nonumber \\
&=&M(g_{N}+\varphi h_{N})-N(g_{M}+\varphi h_{M}) \nonumber \\
&=&M^{t}\displaystyle\frac{\partial g_{N}}{\partial t}+M^{x}\displaystyle\frac{\partial g_{N}}{\partial x}
+\varphi(M^{t}\displaystyle\frac{\partial h_{N}}{\partial t}+M^{x}\displaystyle\frac{\partial h_{N}}{\partial x}) \nonumber \\
&&+M^{\varphi}h_{N}-(N^{t}\displaystyle\frac{\partial g_{M}}{\partial t}+N^{x}\displaystyle\frac{\partial g_{M}}{\partial x} \nonumber \\
&&+\varphi(N^{t}\displaystyle\frac{\partial h_{M}}{\partial t}+N^{x}\displaystyle\frac{\partial h_{M}}{\partial x}) 
+N^{\varphi}h_{M}) \nonumber \\
&=&M^{t}\displaystyle\frac{\partial g_{N}}{\partial t}+M^{x}\displaystyle\frac{\partial g_{N}}{\partial x}
+\varphi(M^{t}\displaystyle\frac{\partial h_{N}}{\partial t}+M^{x}\displaystyle\frac{\partial h_{N}}{\partial x}) \nonumber \\
&+&(g_{M}+\varphi h_{M})h_{N}-N^{t}\displaystyle\frac{\partial g_{M}}{\partial t}-N^{x}\displaystyle\frac{\partial g_{M}}{\partial x} \nonumber \\
&&-\varphi(N^{t}\displaystyle\frac{\partial h_{M}}{\partial t}+N^{x}\displaystyle\frac{\partial h_{M}}{\partial x}) 
-(g_{N}+\varphi h_{N})h_{M} \nonumber \\
&=&(M^{t}\displaystyle\frac{\partial g_{N}}{\partial t}+M^{x}\displaystyle\frac{\partial g_{N}}{\partial x}+g_{M}h_{N}-N^{t}\displaystyle\frac{\partial g_{M}}{\partial t}-N^{x}\displaystyle\frac{\partial g_{M}}{\partial x}-g_{N}h_{M}) \nonumber \\
&+&\varphi(M^{t}\displaystyle\frac{\partial h_{N}}{\partial t}+M^{x}\displaystyle\frac{\partial h_{N}}{\partial x}-N^{t}\displaystyle\frac{\partial h_{M}}{\partial t}-N^{x}\displaystyle\frac{\partial h_{M}}{\partial x}) \,\, ;\nonumber 
\end{eqnarray}
the result follows.
\end{proof}

Let us set
$$
\mathcal H=\{N\in \mathcal G \vert g_{N}=0\}
$$
and
$$
\mathcal J=\{N\in \mathcal G \vert h_{N}=0\}\,\, .
$$
\begin{proposition}$\mathcal J$ is an ideal of $\mathcal G$ and $\mathcal H$ is a subalgebra of $\mathcal G$.
Furthermore $\mathcal G=\mathcal H \oplus \mathcal J$, 
$\mathcal H$ has dimension $6$ and admits 
$(N_{1},...,N_{6})$ as a basis, and $\mathcal J=\{N_{u}\vert u \,\, \text{solution of} \,\, (\mathcal E_{2})\}$.
\end{proposition}
\begin{proof}Clearly, $N\in \mathcal J$ if and only if $C_{1}=...=C_{6}=0$ ; in particular, if $N\in \mathcal J$ then
$N^{t}=N^{x}=0$. Therefore $M\in \mathcal G$ and $N\in \mathcal J$ imply $h_{[M,N]}=0$, \it i.e. \rm $[M,N]\in \mathcal J$:
$\mathcal J$ is an ideal of $\mathcal G$.

Furthermore, $g_{M}=g_{N}=0$ imply $g_{[M,N]}=0$ : $\mathcal H$ is a subalgebra of $\mathcal G$.
The last two assertions clearly hold.
\end{proof}

\section{Some symmetries}

Let $N\in \mathcal G$, let $\kappa\in\mathbf R$, and let $\varphi$ be a solution of $(\mathcal E_{2})$ ; then $e^{\kappa N}$ maps  $(t,x,\varphi,A,B)$ to $(t_{\kappa},x_{\kappa},\varphi_{\kappa},A_{\kappa},B_{\kappa})$ ; setting \begin{eqnarray}\varphi_{\kappa}=\psi_{\kappa}(t_{\kappa},x_{\kappa})\,\, ,\end{eqnarray} it follows that $\psi_{\kappa}$ is also a solution of $(\mathcal E_{2})$.
We shall denote
\begin{eqnarray}
e^{\kappa\tilde{N}}:\varphi \mapsto \psi_{\kappa}
\end{eqnarray}
the associated one--parameter group.
\begin{lemma}If
\begin{eqnarray}
N=N^{t}\displaystyle\frac{\partial}{\partial t}+N^{x}\displaystyle\frac{\partial}{\partial x}+N^{\varphi}\displaystyle\frac{\partial}{\partial \varphi}
 + ... \in \mathcal G \nonumber
 \end{eqnarray}
then
 \begin{eqnarray}
 \tilde{N}(\varphi)=-N^{t}\displaystyle\frac{\partial \varphi}{\partial t}-N^{x}\displaystyle\frac{\partial \varphi}{\partial x}+N^{\varphi} \,\, , \nonumber
 \end{eqnarray}
 that is

\begin{eqnarray}
\tilde{N}(\varphi) 
&&=(C_{1}t^{2}+C_{2}t+C_{3})\displaystyle\frac{\partial \varphi}{\partial t}+(\displaystyle\frac{1}{2}x(2C_{1}t+C_{2})+(C_{4}t+C_{5}))
\displaystyle\frac{\partial \varphi}{\partial x} \nonumber \\
&&+g+\varphi(\displaystyle\frac{\tilde{r}}{2\sigma^{2}}x(2C_{1}t+C_{2})-\displaystyle\frac{C_{1}}{2\sigma^{2}}x^{2}-\displaystyle\frac{C_{4}x}{\sigma^{2}}
-\displaystyle\frac{\tilde{s}^{2}}{2\sigma^{2}}(C_{1}t^{2}+C_{2}t+C_{3}) \nonumber \\
&&+\tilde{r}\displaystyle\frac{C_{4}t+C_{5}}{\sigma^{2}}
+\displaystyle\frac{2C_{1}t+C_{2}}{4}+C_{6}) \,\, . \nonumber 
\end{eqnarray}
\end{lemma}
\begin{proof}
Let us rewrite $(6.1)$ as
$$
(e^{\kappa N}(\varphi))(t,x)=(e^{\kappa \tilde{N}}(\varphi))(e^{\kappa N}(t),e^{\kappa N}(x))\,\, .
$$
Developping at order one in $\kappa$ gives
\begin{eqnarray}
\varphi+\kappa N^{\varphi} \nonumber 
&=&\varphi(t+\kappa N^{t},x+\kappa N^{x})+\kappa \tilde{N}(\varphi)+o(\kappa) \nonumber \\
&=&\varphi+\kappa N^{t}\displaystyle\frac{\partial \varphi}{\partial t}+\kappa N^{x}\displaystyle\frac{\partial \varphi}{\partial x}
+\kappa \tilde{N}(\varphi)+o(\kappa) \,\, ,\nonumber 
\end{eqnarray}
whence the result.
\end{proof}
We shall set
\begin{eqnarray}
C_{N}^{\kappa}(t,S):=(e^{\kappa\tilde{N}}\varphi)(t,\ln(S))
\end{eqnarray}
and
\begin{eqnarray}
C_{j}^{\kappa}:=C_{N_{j}}^{\kappa}
\end{eqnarray}
($1\leq j \leq 6$).

Some of these transforms can actually be explicitly computed ; for instance
\begin{eqnarray}\tilde{N_{3}}(\varphi)=\displaystyle\frac{\partial \varphi}{\partial t}-\frac{\tilde{s}^{2}}{2\sigma^{2}}\varphi
\end{eqnarray}
whence
\begin{eqnarray}
(e^{\kappa\tilde{N_{3}}}\varphi)(t,x)=e^{-\displaystyle\frac{\kappa\tilde{s}^{2}}{2\sigma^{2}}}\varphi(t+\kappa,x) 
\end{eqnarray}
and
\begin{eqnarray}
C_{3}^{\kappa}(t,S)=e^{-\displaystyle\frac{\kappa \tilde{s}^{2}}{2\sigma^{2}}}C(t+\kappa , S)\,\, .
\end{eqnarray}

\begin{eqnarray}
\tilde{N_{4}}(\varphi)=t\displaystyle\frac{\partial \varphi}{\partial x}+(\displaystyle\frac{\tilde{r}t-x}{\sigma^{2}})\varphi
\end{eqnarray}

therefore
\begin{eqnarray}
(e^{\kappa\tilde{N_{4}}}\varphi)(t,x)=e^{\displaystyle\frac{\kappa}{\sigma^{2}}(\tilde{r}t-x)-\displaystyle\frac{\kappa^{2}t}{2\sigma^{2}}}\varphi(t, x+\kappa t) 
\end{eqnarray}
and
\begin{eqnarray}
C_{4}^{\kappa}(t,S)=e^{\displaystyle\frac{\kappa t}{2\sigma^{2}}(2\tilde{r}-\kappa)}S^{-\displaystyle\frac{\kappa}{\sigma^{2}}}C(t, e^{\kappa t} S)\,\, .
\end{eqnarray}
\bigskip
\begin{eqnarray}\tilde{N_{5}}(\varphi)=\displaystyle\frac{\partial \varphi}{\partial x}+\displaystyle\frac{\tilde{r}\varphi}{\sigma^{2}}
\end{eqnarray}
whence
\begin{eqnarray}
(e^{\kappa\tilde{N_{5}}}\varphi)(t,x)=e^{\displaystyle\frac{\kappa\tilde{r}}{\sigma^{2}}}\varphi(t,x+\kappa)\,\, ,
\end{eqnarray}
and
\begin{eqnarray}
C_{5}^{\kappa}(t,S)=e^{\displaystyle\frac{\kappa \tilde{r}}{\sigma^{2}}}C(t,e^{\kappa}S)\,\, .
\end{eqnarray}

\bigskip

Lastly,
\begin{eqnarray}
\tilde{N_{6}}(\varphi)=\varphi
\end{eqnarray}
whence
\begin{eqnarray}
(e^{\kappa\tilde{N_{6}}}\varphi)(t,x)=e^{\kappa}\varphi(t,x)
\end{eqnarray}
and
\begin{eqnarray}
C_{6}^{\kappa}(t,S)=e^{\kappa}C(t,S)\,\, .
\end{eqnarray}

\begin{corollary}
Let $C$ denote a solution of $(\mathcal E)$.
Then, for each $\kappa\in \mathbf R$, the $C_{i}^{\kappa}$($3\leq i \leq 6$) defined by $(6.7)$, $(6.10)$,
$(6.13)$ and $(6.16)$ are also solutions of $(\mathcal E)$.
\end{corollary}

Together the transformations described in $(6.7)$ and $(6.16)$ come from the invariance of the original equation
via multiplication of the solution by a scalar and translation in time ; $(6.13)$ then comes from the homogeneity in $S$ --- these
could be expected. The transformation given by $(6.10)$  is, however, not so easy to understand.
It would be interesting to find some financial interpretation for it.

\section{Acknowledgements}

I presented prior versions of this work at the Ascona Conference (May 2008), at the seminar of
the \it Grupo de F\'\i sica--Matem\'atica \rm (Lisbon, July 2008), at the \it Stochastic Analysis Seminar \rm
(Loughborough, January 2010) and at the \it S\'eminaire de Sciences Actuarielles \rm (\it Universit\'e
Libre de Bruxelles\rm, February 2010). For these invitations I am deeply indebted to (respectively)
Professors Robert Dalang, Marco Dozzi and Francesco Russo, Professors Ana Bela Cruzeiro and Jean--Claude Zambrini,
Professor J\'ozsef L\"orinczi, and Professor Pierre Patie. I am grateful to many members of the various audiences,
notably Professors Eckhard Platen and Eric Carlen, for their remarks.
I also thank Professor Archil Gulisashvili for his remarks on the original text.

\bibliographystyle{amsalpha}

\end{document}